\documentclass[conference]{IEEEtran}
\IEEEoverridecommandlockouts
\usepackage{cite}
\usepackage{amsmath,amssymb,amsfonts}
 \usepackage{mathrsfs}
\usepackage{graphicx}
\usepackage{textcomp}
\usepackage{xcolor}
\usepackage{bm}
\usepackage{bbm}
\usepackage{ulem}

\usepackage{algorithm}
\usepackage{algpseudocode}
\usepackage{graphics}
\usepackage{epsfig}
\usepackage{verbatim}
\newtheorem{myDef}{Definition} 
\newtheorem{myTheo}{Theorem}

\newtheorem{proof}{Proof}
\newcommand{\RNum}[1]{\uppercase\expandafter{\romannumeral #1\relax}}

\def\BibTeX{{\rm B\kern-.05em{\sc i\kern-.025em b}\kern-.08em
    T\kern-.1667em\lower.7ex\hbox{E}\kern-.125emX}}
\begin{document}
\title{A Pre-Allocation Design for Cost Minimization and Delay Constraint in Vehicular Offloading System\\
}
\author{
\IEEEauthorblockN{Zhijie Chen\IEEEauthorrefmark{1}\IEEEauthorrefmark{2}, Bo Yang\IEEEauthorrefmark{2}, Cailian Chen\IEEEauthorrefmark{2}, Xinping Guan\IEEEauthorrefmark{2}~\IEEEmembership{Fellow,~IEEE}} 
  
\IEEEauthorblockA{\IEEEauthorrefmark{1}Department of Computer Science and Engineering, Shanghai Jiao Tong University, Shanghai, China }  
\IEEEauthorblockA{\IEEEauthorrefmark{2}Key Laboratory of System Control and Information Processing, Ministry of Education of China, China } 
}

\maketitle

\begin{abstract}
To accommodate exponentially increasing traffic demands of vehicle-based applications, operators are utilizing offloading as a promising technique to improve quality of service (QoS), which gives rise to the application of Mobile Edge Computing (MEC). While the conventional offloading paradigms focus on delay and energy tradeoff, they either fail to find efficient models to represent delay, especially the queueing delay, or underestimate the role of MEC Server. In this paper, we propose a novel \textbf{P}re-\textbf{A}llocation \textbf{D}esign for vehicular \textbf{O}ffloading (\textbf{PADO}). A task delay queue is constructed based on an allocate-execute separate (AES) mechanism. Due to the dynamics of vehicular network, we are inspired to utilize Lyapunov optimization to minimize the execution cost of each vehicle and guarantee task delay. The MEC Server with energy harvesting devices is also taken into consideration of the system. The transaction between vehicles and server is decided by a Stackelberg Game framework. We conduct extensive experiments to show the property and superiority of our proposed framework.
\end{abstract}


\section{Introduction}
With the ever-increasing number of vehicles on the roads and the development of automobile industry, vehicles have been a significant component of the mobile devices connecting to the Internet. Nowadays, vehicles can support various mobile applications, such as image-aided navigation\cite{vu2012real} and vehicular augmented reality\cite{qiu2017augmented}. These applications require huge quantities of computation resources. The increasing needs for resources along with the pursuit of higher performance for advanced vehicular applications poses a great challenge to run computationally intensive applications on the resource constrained vehicles.

The more rigorous requirements for mobile devices on execution and storage ability calls for a new paradigm called Mobile Edge Computing (MEC) that relies on the internal wireless network to acquire computational capabilities\cite{xu2013survey}. Compared with the traditional local execution of tasks, offloading workload to an MEC Server not only reduces the execution delay but also saves the energy consumption of users. Therefore, it is pressing to take advantage of the MEC features to share the intensive workload of users\cite{liu2014effective}. 

A number of works have studied offloading schemes in an MEC system. Cordeschi $et$  $al.$ \cite{cordeschi2014reliable} designed a distributed and adaptive traffic offloading scheme for cognitive cloud vehicular networks. Sasaki et al. \cite{sasaki2016vehicle} designed an infrastructure-based vehicle control system to reduce latency and balance computational load. Wang $et$  $al.$ \cite{wang2017computational} studied computational offloading problem to optimize the total consumption cost incurred by the usage of the limited computational resources. Du $et$  $al.$\cite{du2018computation} considers a cognitive vehicular network and formulate a dual-side optimization problem to minimize vehicle side and server side cost. 

However, as the computation resources owned by vehicles are limited, not all the tasks can be disposed at once. More chances are that tasks go under a queueing stage and pending execution under an FIFO principle. For each vehicle, multiple tasks generated in different moments coexist, forming a task backlog queue. Most of the aforementioned works either ignore the queueing delay or approximate delay estimation, leading to defective and unreliable solutions. 

In addition, it is worthwhile to leverage energy harvesting  technology to capture the green energy (e.g., solar, wind and solar radiation, etc.) for charging battery constantly. Rather than the design adopted in \cite{mao2016dynamic,zhang2018energy}, energy harvesting devices in our settings are equipped by the MEC Server battery, which can embrace the benefit of massive deployment.
The integration of renewable energy is provides another option for the MEC Server other than the battery storage and power grid. At each time slot, the server can alternate the source to provide energy for its service while preventing the battery level overload. Thus the cost minimization of MEC Server is also concerned in our framework.

As comparison to the state-of-the-art works, our proposed offloading scheme explicitly takes into account the key aspects specific in vehicular network. The contributions of our work can be summarized as follows:
\begin{itemize}
\item We propose a novel \textbf{P}re-\textbf{A}llocation \textbf{D}esign for vehicular \textbf{O}ffloading (\textbf{PADO}) framework, which enables the vehicle to carry out allocate-execute separate mechanism on deciding its offloading strategy. Under the PADO framework, a delay queue which directly represents the task queueing delay is formulated. The design of delay queue can facilitate arbitary delay deadline and tasks with stratified deadline requirement.

\item The execution cost minimization problem, which is an intractable high-dimensional Markov decision problem, is formulated by a low-complexity online Lyapunov optimization based scheduling framework. The usage of Lyaponov Optimization changes the stochastic problem into a deterministic one and its solution only depends on non-stochastic variables. Both parties in the PADO framwork, a.k.a vehicles and MEC Server, adopt this framework for balancing queues stability and cost minimization.

\item A one-leader and multi-follower Stackelberg Game between vehicles and MEC Server is formulated to characterize coupled trading behavior between players. The solution to the dual-side optimization problem gives us indicators for the current vehicle's workload. A closed-form unit task price is derived, which can be used to guide vehicle's offloading action. 

\item We conduct extensive experiments to show the property and superiority of our proposed framework. Theorems about the stability of queues are verified by simulation results. Moreover, the effectiveness of the proposed algorithm is demonstrated by comparisons with three benchmark policies. It is shown that the PADO framework enables vehicles to fully utilize the execution deadline and improves in terms of execution cost. 
\end{itemize}

The rest of the paper is organized as follows. We introduce the system model in Section \RNum{2} and formulate the optimization problem for vehicles and MEC Server in Section \RNum{3}. In Section \RNum{4}, we use Stackelberg Game to determine the offloading strategy in a specific time slot. We evaluate our proposed offloading scheme with simulation in Section \RNum{5}. The paper is concluded in Section \RNum{6}.

\section{System Model}
\subsection{Mobile Edge Computing System}
Typically, a mobile edge computing system involves several vehicles and road side units. We denote $\mathcal{M} = \{1,2,...,M\}$ as the set of vehicles, where $M$ is the number of vehicles. These vehicle users are the producer of the whole system, which means they will stochastically generate certain amount of tasks at the beginning of each time slot. Here the stochasticity is two-fold, i.e., the occurrence of computation task and its numeric amount are both randomly distributed. In the light of the bursting nature of task arrival, vehicles either produce tasks, the quantity of which has a lower bound, or do not produce any new task at all. We thus use an i.i.d Bernoulli distribution across vehicles to model this burst arrival. We define $\rho$ as the parameter of the distribution, a.k.a the arrival rate. If $\xi_i^t$ is denoted to indicate the occurrence of a newly-arrived task on slot $t$ in vehicle $i$, then $p(\xi_i^t=1) = 1-p(\xi_i^t=0) = \rho$. The numeric quantity of a computation task, denoted as $R_{task}$ is uniformly distributed on the closed interval $[R_{min},R_{max}]$. So in any time interval, the actual amount of newly generated tasks $R_i(t)$ can be defined as below,
\begin{equation}
    R_i(t) = \xi_i^t \cdot R_{task}^t
\end{equation}

Other than the quantity, a computation task is attributed by other factors. The most important one is the maximal tolerant delay. To guarantee the user Quality of Service (QoS), applications with high requirement in timeliness need a response within certain time after generation. We denote $\tau_i^t$ as the time limit of vehicle $i$'s task in slot $t$.

The execution of computation task is based on physical hardware (such as GPU, memory and HDD). In actual settings, different tasks also vary in the requirement of hardware. For instance, we call an application that has high requirement in GPU configuration as a GPU-intensive task. Vehicle-based virtual reality is one of the representative applications of this kind. So in the characterization of a computation task, we need to identify the request of resources respectively so that the edge servers can organize the VMs. Without loss of generality, we denote $k$ as the number of total types of resources, and then $ \mathcal{r}=\{r_1,r_2,...,r_k\}$ is the set of configuration on hardware, where $r_i, 1 \le i \le k$ is the configuration on type-$i$ resource. Thereafter, we use a tuple $A_i^t = (R_i(t),\tau_i^t,\mathcal{r})$ to characterize a computation task based on the discussion above.

\subsection{Computation Model on Vehicles and Servers}
We assume that each computation task can be chopped into independent parts at any ratio, for local execution and offloading. This partition can be implemented by Spark~\cite{zaharia2010spark}. We denote $\alpha_i^t \in [0,1]$ as the percentage of locally computed task, and $\beta_i^t \in [0,1]$ as the offloading percentage. For vehicle $i$, it will decide $\alpha_i^t$ and $\beta_i^t$ for each task considering delay cost and the accessibility to the server.

The delay cost is comprised of two components, the queueing delay and service time a.k.a the execution time. As discussed above, vehicles can either execute the task locally, or offload it to the remote server. For the local part, the frequency scheduled for the task $A_i^t$ is denoted as $f_{i,t}^{local}$, which can be implemented by adjusting the chip voltage with DVFS techniques\cite{rabaey2002digital}. The execution time can be expressed as
\begin{equation}
    D_{i,t}^{local} = \frac{\alpha_i^t R_i^t}{f_{i,t}^{local}}
\end{equation}

Compared with the server, vehicle users have limited local resources. So chances are that tasks are not instantly executed just on its arrival, but will have to be waiting in a backlog queue with FIFO discipline. This will generate the queueing delay. To quantify the queueing delay, we propose a novel delay time indicating method.

First, to facilitate the downstream method, we make a modification for the maximum tolerant delay $\tau_i^t$, which is a continuous variable. Here binning method is used to discretize $\tau_i^t$.  $\mathbf{\Gamma} = \{\Gamma_1,\Gamma_2,...\Gamma_s\}$ is denoted as the set of typical delay constraints, where $\Gamma_1 \le \Gamma_2 \le ...\le \Gamma_s$ and $s$ is the number of delay constraints. As the task with lower delay requirement is more sensitive to the delay (e.g. the task with time limit of 10s is more sensitive to any 1s increase compared to that of time limit of 1 minute), typically we set the $\mathbf{\Gamma}$ to be a geometric sequence. We assign a delay index $h_i^t$ for each of the task according to its $\tau_i^t$ as below.

\begin{equation}
    h_i^t = \mathop{\arg\max}_{j \in [1,s]}\{\tau_j < \tau_i^t\}
\end{equation}

$h_i^t$ can also be deemed as the priority index of a task. The lower $h_i^t$ is, the quicker response for the result should be. After discretization of $\tau_i^t$ into $s$ delay constraints, we can formulate a delay time sequence denoted by $Q_{s,i}^t$ for each $\Gamma_s$, where its update equation is 
\begin{equation}
    Q_{s,i}^{t+1} = \max\{Q_{s,i}^{t} - \zeta,0\} + \frac{\alpha_i^t \mathbbm{1}_{\{h_i^t=s\}} R_i^t }{f_{s}^{local}}
\end{equation}
where $\zeta$ is the time interval of each slot, $\mathbbm{1}_{\{h_i^t=s\}}$ is an indicator function that maps to 1 if $h_i^t=s$, and 0 otherwise. We also assume the CPU-cycle frequency is constrained by $f_{s}^{max}$.The queueing delay for tasks at time slot $t$ is the delay before it is generated, namely $Q_{s,i}^{t}$. Notice that the allocation and execution are separate in our settings. In other words, $\alpha_i^t$ and $f_s^{local}$ are designated to a task when it arrives, while the true execution may begin when the task is at the head of the queue. That's why the framework is called pre-allocation here.

Another part of the task is offloaded to the MEC server. This process can be viewed as a renting behavior for the virtual machines from server. Typically, the server organizes and provides its service in the form of virtual machine, which contains quantified computation resources. These VMs work independently, and serve as a Plug and Play application. It means when the task is scheduled to the server with resources not exceeding its limit quantum, the task can be executed immediately without any queueing or waiting. Hence the delay of  offloaded tasks can be expressed as

\begin{equation}
    D_{i,t}^{server} = \frac{\beta_i^t R_i^t}{f_{i,t}^{server}}
\end{equation}

Therefore, the total delay $D_i^t$ can be expressed as 

\begin{equation}
    D_{i,t} = \left\{
    \begin{array}{lcl}
    \frac{R_i^t}{f_{i,t}^{server}}& &{\alpha = 0}\\
    \max \{Q_{s,i}^{t-1} + D_{i,t}^{local} ,D_{i,t}^{server}\} & &{\alpha \neq 0}

    \end{array} \right.
\label{eq6}
\end{equation}

One of the advantages of using MEC server for offloading is that the delay can be reduced. We assume that the computation resource on server can easily satisfy the deadline requirement of any task, say, $\tau_d$, i.e.

\begin{equation}
    \frac{R_i^t}{f_{i,t}^{server}} < \tau_d
\end{equation}

Since $D_{i,t}^{server}<\frac{R_i^t}{f_{i,t}^{server}}$, we can also get $D_{i,t}^{server}<\tau_d$. According to Eq. \ref{eq6}, if we want to guarantee the total delay bound, the following inequality must be satisfied:

\begin{equation}
    Q_{s,i}^{t} < \tau_d
\end{equation}

Apart from the delay, vehicles also pay attention to the energy cost for the task executed locally. The energy consumption model is describe below. For task $A_i^t$, $\alpha_i^t$ of the total amount is locally executed at CPU frequency $f_{i,t}^{local}$, so the energy cost throughout its execution can be expressed as
\begin{equation}
    E_i^t = \kappa (f_{i,t}^{local})^2 \frac{\alpha_i^t R_i^t}{f_{i,t}^{local}} = \kappa \alpha_i^t f_{i,t}^{local} R_i^t
\end{equation}
where $\kappa$ is the effective switched capacitance that depends on the chip architecture \cite{burd1996processor}.

\subsection{Energy Model on Server}
In this paper, the MEC server is equipped with a renewable energy generator, which consistently provides energy supply for the offloading system\cite{yang2016distributed}. We denote the production rate as $U^t \in [0, U_{max}]$. The server benefits from the generation of renewable energy by reducing the amount of buying energy from the central electric grid. Compared to that from the latter one, the cost of renewable energy is negligent in this scenario, exempt from the investment of raw materials as well as the depravity of long-distance transmission. To make the best of the harvested energy, the server is also equipped with a battery as a buffer to store the energy. We denote the charging level of the battery in slot $t$ as $B^t$.

For the MEC server, it supplies energy for the tasks offloaded from vehicle users. For task $A_i^t$, we have a total amount of $(1-\alpha_i^t) R_i^t$ with allocated CPU frequency $f_{i,t}^{server}$. So to empower the rent virtual machine, energy consumption for task $A_i^t$ can be expressed as
\begin{equation}
    N_i^t = \kappa (f_{i,t}^{server})^2 D_{i,t}^{server}= \kappa \beta_i^t f_{i,t}^{server} R_i^t
\end{equation}
To avoid charging deficiency, the generated renewable energy is firstly supplied to execute the offloaded tasks. If there is redundant energy, then the surplus will be charged into the battery. Let $C^t$ be the current charging rate to the battery and thus $C^t \in [0, \min \{C^{max},(U^t-\sum_{i=1}^M N_i^t)^+\}]$, where $(x)^+ \triangleq \max\{0,x\}$. Otherwise, when the instant energy consumption exceeds the amount of renewable energy generation, the energy supply gap will be filled by either the battery or by the power grid. We denote the amount from latter resource as $G^t$. And so $(\sum_{i=1}^M N_i^t-U^t)^+-G^t$ is the discharged amount of energy from battery under this condition.

With the notations stated above, we can derive the battery dynamics on the basis of charging and discharging strategy, 
\begin{equation}
    B_i^{t+1}=B_i^t-\eta^-((\sum_{i=1}^M N_i^t-U^t)^+-G^t) + \eta^+ C^t
\end{equation}
where $0<\eta^+\le1$ and $\eta^-\ge1$ represent the charging and discharging efficiency respectively. And $\eta^-((\sum_{i=1}^M N_i^t-U^t)^+-G^t)$ denotes the amount of energy extracted from the battery, so it must satisfies:
\begin{equation}
   0 \le \eta^-((\sum_{i=1}^M N_i^t-U^t)^+-G^t) \le B_i^t
\label{eq12}
\end{equation}
which means the extracted energy must not exceed the current battery storage.

The backlog of tasks in vehicular users and the energy storage makes the offloading decisions for users and energy supply policy for servers to be more complicated, compared to the conventional mobile edge computing systems. These two kinds of backlog entail temporally correlated workload and energy level and makes the system decisions coupled in different time slots. In other words, the decisions are not myopic any more but to make a tradeoff between the long term and current profit.

\section{Problem Formulation}
\subsection{Cost minimization for vehicular users}
The vehicular users are motivated to minimize the time-average expected cost for executing the tasks with guaranteed delay. This problem can be formulated as below,

\begin{align*}
    \mathbf{P1}: & \min  \quad U_i =\lim_{T \to +\infty} \sum_{t=0}^{T-1}\mathbb{E}[\sum_{s = 1}^S E_i^t + \frac{\beta_i^t R_i^t}{f_{i,t}^{server}}g(f_{i,t}^{server}) \\
    &~~~~~~~~~~~~~~~~~~~~~~~~~~~~+ (1-\alpha_i^t-\beta_i^t)\Upsilon]\\
    & \begin{array}{r@{\quad}c@{}}
        s.t.\quad C1:& Q_{i,t} < \tau_d\\
            \quad C2: & 0 \le f_s^{local} \le f_s^{max}\\
            \quad C3: & 0 \le \alpha_i^t \le 1\\
            \quad C4: & 0 \le \beta_i^t \le 1\\
            \quad C5: & \alpha_i^t + \beta_i^t \le 1\\
       \end{array}
\end{align*}
where $\sum_{s = 1}^S E_i^t$ denotes the overall energy consumption for local execution. And $\frac{\beta_i^t R_i^t}{f_{i,t}^{server}}g(f_{i,t}^{server})$ denotes the payment for VMs from MEC server. As a general business model, the service provided by server is charged by time. $g(f_{i,t}^{server})$ denotes the server's unit price for providing CPU frequency $f_{i,t}^{server}$. The expectation above is with respect to the potential randomness of the control policy. $\Upsilon$ is the drop loss of a task. If the task is neither executed locally nor offloaded, it can be deemed as dropped (or deployed to cloud), the price of which is $\Upsilon$.

One challenge in solving $\mathbf{P1}$ is due to constraint $C1$, which brings time correlation to the problem. We leverage tools from Lyapunov optimization framework. First, we are going to modify $C1$, and transform the off-line problem $\mathbf{P1}$ to an online optimization problem. To guarantee the delay of task execution, we construct a virtual queue $W_i^t$, which can be expressed as follow:
\begin{equation}
    W_{s,i}^{t+1} = \max \{ W_{s,i}^t + Q_{s,i}^{t+1} - \Gamma_s,0 \}
\end{equation}

We can define the Lyapunov function as $L_i(t) = \frac{1}{2}\sum_{i=1}^S (W_{s,i}^t)^2$. Then, the conditional Lyapunov drift-plus-penalty for slot $t$ is given by:
\begin{equation}
    \Theta(L_i(t)) = \mathbb{E}[L_i(t+1)-L_i(t)+ V\cdot U_i^t|W_i(t)]
\end{equation}
where $V$ is a user-determined hyperparameter.
While we assume that in any time slot, at most one task is newly generated, whose delay priority is denoted as $\hat{s}$. Then the drift-plus-penalty can be slacked with an upper bound,

\begin{align*}
    & \Theta(L_i(t))+V\cdot U_i  \\ &=\sum_{s=1}^S\frac{1}{2}[(W_{s,i}^{t+1})^2-(W_{s,i}^t)^2] +V[\sum_{s = 1}^S E_i^t + \frac{\beta_i^t R_i^t}{f_{i,t}^{server}}g(f_{i,t}^{server})] \\
    &\le\sum_{s=1}^S (\Gamma_s-\frac{\alpha_i^t \mathbbm{1}_{\{h_i^t=s\}} R_i^t }{f_{s}^{local}})W_{s,i}^t + \frac{1}{2}\Gamma_s^2-\Gamma_s\frac{\alpha_i^t \mathbbm{1}_{\{h_i^t=s\}} R_i^t }{f_{s}^{local}} \\
    &\quad +\frac{1}{2}(Q_{s,i}^{t+1})^2 +V[\kappa\beta_i^tf_s^{local}R_i^t+\frac{\beta_i^t R_i^t}{f_{i,t}^{server}}g(f_{i,t}^{server})]\\
    & = B + (|Q_{\hat{s},i}^t - \zeta|+W_{\hat{s},i}^t-\Gamma_{\hat{s}})\frac{\alpha_i^t R_i^t}{f_{\hat{s}}^{local}}+\frac{1}{2}(\frac{\alpha_i^t R_i^t}{f_{\hat{s}}^{local}})^2\\
    &\quad +V[\kappa\beta_i^tf_s^{local}R_i^t+\frac{\beta_i^t R_i^t}{f_{i,t}^{server}}g(f_{i,t}^{server})]
\end{align*}
where $B$ is a constant with regard to the function with variables $f_s^{local}$, $\alpha_i^t$ and $\beta_i^t$.
\begin{proof}
See Appendix A.
\end{proof}

At every time slot, the vehicle will make decisions of the local allocated CPU frequency and the proportion to be offloaded for controlling the upper bound of cost. To this end, the Lyapunov framework minimizes the right hand side of the drift-plus-penalty expression. The optimization problem for vehicle $i$ can thus be formulated as

\begin{align*}
    \mathbf{P2}:\min \quad &\widetilde{U}_i =  (|Q_{\hat{s},i}^t - \zeta|+W_{\hat{s},i}^t-\Gamma_{\hat{s}})\frac{\alpha_i^t R_i^t}{f_{\hat{s}}^{local}}+\frac{1}{2}(\frac{\alpha_i^t R_i^t}{f_{\hat{s}}^{local}})^2\\
    &+V[\kappa \beta_i^tf_s^{local}R_i^t+\frac{\beta_i^t R_i^t}{f_{i,t}^{server}}g(f_{i,t}^{server})+(1-\alpha_i^t-\beta_i^t)\Upsilon] \\
    s.t. & \quad C2,\quad C3,\quad C4,\quad C5
\end{align*}

\subsection{Revenue Maximization for MEC Server}
As for MEC server, it provides high-quality and ultra-low latency service for vehicles. Its objective is to make a profit for providing such service. However, it also pays close attention to the remnant in battery. For a sustainable business mode, the battery should be in a healthy and stable across time.  The revenue maximization problem can be defined as below

\begin{align*}
    \mathbf{P3}: \quad& \max \sum_{i=1}^N \frac{\beta_i^t R_i^t}{f_{i,t}^{server}}g(f_{i,t}^{server}) - \chi^t G^t \\
    s.t.\quad & (7),(11),(12), \\
    & B_i^t < \infty, \quad  \forall  t\quad 1 \le t \le T
\end{align*}
where $\chi^t$ is the unit price of purchasing electricity from power grid. So $\chi^t G^t$ is the overall payment for extra power when the renewable energy supply is insufficient. 

Nevertheless, due to the energy causality constraint (\ref{eq12}), the decision making process of $G^t$ is coupled among different time slots. To facilitate further analysis, we first introduce an upper bound $E_{max}$  for the discharged energy, and the problem can be rewritten as the following
\begin{align*}
    \mathbf{P4}: \quad& \Omega = \max \sum_{i=1}^N \frac{\beta_i^t R_i^t}{f_{i,t}^{server}}g(f_{i,t}^{server}) - \chi^t G^t \\
    s.t.\quad & (7),(11), \\
    & \eta^-((\sum_{i=1}^M N_i^t-U^t)^+-G^t) \le E_{max} ,\\
    & B_i^t < \infty, \quad  \forall  t\quad 1 \le t \le T
\end{align*}

As will be elaborated later, the proposed solution to $\mathbf{P4}$ also satisfies constraint ($\ref{eq12}$). Next we define the perturbation parameter $\theta$ and virtual energy queue $\widetilde{B}^t$ respectively,
\begin{equation}
    \theta \ge \frac{H\chi^t}{\eta^-}+E_{max}
\label{eq15}
\end{equation}
\begin{equation}
    \widetilde{B}^t = B^t - \theta
\end{equation}
where $H > 0$ is a positive control parameter. With regard to the stability of battery queue, we adopt a similar Lyapunov framework on MEC server. The Lyapunov drift-plus-penalty is defined as
\begin{align*}
    &\Delta(J(t))-H \cdot \Omega \\
    &= \frac{1}{2}( (\widetilde{B}^{t+1})^2-(\widetilde{B}^{t})^2)-H(\sum_{i=1}^N \frac{\beta_i^t R_i^t}{f_{i,t}^{server}}g(f_{i,t}^{server}) - \chi^t G^t)\\
    &\le C + (\widetilde{B}^{t}\eta^-+H\chi^t)G^t\\
    &\quad -\sum_{i=1}^M\beta_i^tR_i^t(\kappa\widetilde{B}^{t}\eta^-f_{i,t}^{server}+\frac{H}{f_{i,t}^{server}}g(f_{i,t}^{server}))
\end{align*}

To minimize the RHS of the above inequality, the MEC server make decisions of $G^t$ and $f_{i,t}^{server}$. So the revenue maximization of MEC server can now be expressed as 
\begin{align*}
    \mathbf{P5}: \min & \quad  (\widetilde{B}^{t}\eta^-+H\chi^t)G^t
    -\sum_{i=1}^M\beta_i^tR_i^t(\kappa\widetilde{B}^{t}\eta^-f_{i,t}^{server}\\
    & +\frac{H}{f_{i,t}^{server}}g(f_{i,t}^{server}))\\
    s.t.& \quad (7),(11)
\end{align*}

Thus, in this vehicular edge computing scenario, we aim at devising a bidirectional pricing and energy management scheme for both vehicle users and MEC server. Meanwhile, we expect to guarantee the long term profit of all the agents in the system. To this end, a Stackelberg Game based algorithm is proposed for the agents' own profit maximization respectively. 

\section{A Stackelberg Game Approach}
In this section, we develop a Stackelberg Game model to analyze the offloading mechanism between vehicular users and MEC server according to optimization problem $\mathbf{P2}$ and $\mathbf{P5}$ respectively. First, we formally define the game played under such scenario over the $T$ slots. This game contains vehicular users and MEC server. The bidirectional pricing scheme set by MEC server will variously impact the offloading decision and local allocation for vehicles, which will conversely affect the planning of mechanism of MEC server through its total revenue from users. This leads to a typical instance of the Stackelberg Game, where the MEC server works as a leader and the $M$ vehicular users are the followers who are subject to the decision made by leader. Here we give the complete modeling of this one-leader and multi-follower Stackelberg Game as below,

\begin{myDef}
Vehicular Offloading Stackelberg Game
\noindent $\textbf{Players}$: $M$ vehicular users and  1 MEC server.

\noindent $\textbf{Strategies}$: Each vehicular user  $m \in \mathcal{M}$ determines its own strategy $\mathbf{x_m}=\{\alpha_m^t,\beta_m^t,f_s^{local}\}$, which is a combination of the partition of task $A_m^t$ and the locally allocated CPU frequency to meet the demand of its long term profit. The MEC server also makes decision at every slot $\mathbf{y} = \{\mathbf{f}_{t}^{server}, G^t\}$, where $\mathbf{f}_{t}^{server} = \{f_{1,t}^{server},f_{2,t}^{server},...,f_{M,t}^{server}\}$ is a CPU frequency allocation vector. And it also decides the amount of energy to buy from power grid.

\noindent $\textbf{Payoff}$: Vehicular users benefit from offloading tasks by saving the energy as well as guaranteeing the delay constraint for each task. MEC server makes profit through payment for renting virtual machines to the vehicular users who offload their tasks. 
\end{myDef}

As stated in Section \uppercase\expandafter{\romannumeral3}, vehicles and server have their own utility functions, which are also mutually correlated. So seeking the best strategy for each of them is equivalent to optimizing the utility functions of vehicle users and MEC server sequentially.

We first address the strategy for vehicular users. Notice that there is variable coupling due to the constraint $C5$ in $\mathbf{P2}$. Hence we employ the Lagrange dual method. 
the optimization function in $\mathbf{P2}$ is a concave function with regard to variable $\alpha_i^t$ and $\beta_i^t$. The Lagrangian relaxation for $\mathbf{P2}$ is defined as

\begin{align*}
    \mathcal{L}(\alpha_i^t,\beta_i^t,\lambda_i^t) =& \widetilde{U}_i^t -\lambda_i^t(\alpha_i^t + \beta_i^t - 1) \\
    =&\lambda_i^t + \frac{1}{2}(\frac{R_i^t}{f_{\hat{s}}^{local}})^2(\alpha_i^t)^2 - [\lambda_i^t+V\Upsilon \\
   & -(|Q_{\hat{s},i}^t -       \zeta|-W_{\hat{s},i}^t-\Gamma_{\hat{s}})\frac{ R_i^t}{f_{\hat{s}}^{local}}]\alpha_i^t\\
    &-(V[\kappa f_s^{local}R_i^t-\frac{ R_i^t}{f_{i,t}^{server}}g(f_{i,t}^{server})+\Upsilon]+\lambda_i^t)\beta_i^t \\
\end{align*}
where $\lambda_i^t$ is the Lagrangian multiplier for constraint $C5$. Due to the convexity of the objective function, the minimum of $\mathcal{L}(\alpha_i^t,\beta_i^t,\lambda_i^t)$ can be derived by the Karush-Kuhn-Tucker (KKT) conditions described as below:

The optimization problem has a local minimum iff. there exists a unique $\lambda_i^t$ s.t.
\begin{align*}
    \nabla_{\alpha_i^t} \mathcal{L}(\alpha_i^t,\beta_i^t,\lambda_i^t) = 0\\
    \nabla_{\beta_i^t} \mathcal{L}(\alpha_i^t,\beta_i^t,\lambda_i^t) = 0\\
    \lambda_i^t \ge 0\\
    \lambda_i^t (\alpha_i^t + \beta_i^t - 1) = 0\\
\end{align*}

We can derive the corresponding $\alpha_i^t$ and $\beta_i^t$ as 

\begin{equation}
    \alpha_i^t = \left\{
    \begin{array}{lcl}
    0&,{\Psi_{i,t}^{loc}>\min\{\Psi_{i,t}^{off},\Psi_{i,t}^{cld}\}}\\
    \min\{1,V\frac{(f_s^{local})^2}{R_i^t}&(\Psi_{cld}-\Psi_{loc})\}\\& ,{\Psi_{i,t}^{loc}<\Psi_{i,t}^{cld}<\Psi_{i,t}^{off}}\\
\min\{1,V\frac{(f_s^{local})^2}{R_i^t}&(\Psi_{off}-\Psi_{loc})\}\\ &,{\Psi_{i,t}^{loc}<\Psi_{i,t}^{off}<\Psi_{i,t}^{cld}}
    \end{array} \right.
\label{alpha}
\end{equation}

\begin{equation}
    \beta_i^t =\left\{
    \begin{array}{lcl}
     0& &otherwise\\
     1-\alpha_i^t& &{\Psi_{i,t}^{loc}<\Psi_{i,t}^{off}<\Psi_{i,t}^{cld}}\\
     1 & &{\Psi_{i,t}^{off} < \max\{\Psi_{i,t}^{loc},\Psi_{i,t}^{cld}\}}\\
    \end{array}\right.
\label{beta}
\end{equation}
where $\Psi_{i,t}^{loc}$,$\Psi_{i,t}^{off}$,$\Psi_{i,t}^{cld}$ can be deemed as the unit price for executing the task for local computation, offloading to MEC server and dropping to cloud respectively. More specifically, these values are derived from the KKT condition and can be expressed as below,
\begin{equation}
    \Psi_{i,t}^{loc} = \frac{1}{V f_s^{local}} (|Q_{s,i}^t-\zeta|+W_{s,i}^t-\Gamma_s)+\kappa f_s^{local}
\end{equation}
\begin{equation}
    \Psi_{i,t}^{off} = \frac{g(f_{i,t}^{server})}{f_{i,t}^{server}}
\end{equation}
\begin{equation}
    \Psi_{i,t}^{cld} = \Upsilon
\end{equation}
\begin{proof}
    See Appendix B.
\end{proof}

With the two variables $\alpha_i^t$ and $\beta_i^t$ selected in (\ref{alpha}) and (\ref{beta}), the vehicle can further determine the pre-allocation amount of local CPU frequency by solving the problem in $\mathbf{P2}$ with regard to variable $f_s^{local}$.

Now that the vehicular users' decision space $\mathbf{x_m}$ is determined, MEC server, as the leader in the Game, will modify its decision $\mathbf{y}$ on the basis of $\mathbf{x_m}$. On scrutinizing the optimization problem for server, $\mathbf{P5}$ can be decomposed into two subproblems
\begin{align*}
    \mathbf{P5}(a):  \min \quad (\widetilde{B}^{t}\eta^-+H\chi^t)G^t
\end{align*}

\begin{align*}
    \mathbf{P5}(b): \min \quad & -\sum_{i=1}^M\beta_i^tR_i^t(\kappa\widetilde{B}^{t}\eta^-f_{i,t}^{server}
     +\frac{H}{f_{i,t}^{server}}g(f_{i,t}^{server}))\\
    s.t.\quad & (7) 
\end{align*}

The objective of $\mathbf{P5(a)}$ is a linear function. So the purchasing strategy $G^t$ can be selected as,
\begin{equation}
     G^t = \left\{
    \begin{array}{ll}
    \max\{(\sum_{i=1}^M N_i^t-U^t)^+ -E_{max},0\} &{(B^t-\theta)\eta^- + H\chi^t>0}\\
    {(\sum_{i=1}^M N_i^t-U^t)^+}&{(B^t-\theta)\eta^- + H\chi^t\le 0}
    \end{array} \right.
\label{eq23}
\end{equation}

In accord with our intuition, MEC server will supply task computation with its battery energy when it has abundant storage. On the other hand, purchase from power grid mounts when the battery is thirsty. We can prove that with adopting this strategy, the battery storage will stabilize within a fixed interval.
\begin{myTheo}
    Under the given strategy in (\ref{eq23}), the battery energy level of MEC server is confined within $[E_{max},\theta-\frac{1}{\eta^-}H\chi^t+\eta^+C^{max}]$, $\forall{t} \in \tau$.
\end{myTheo}
\begin{proof}
See Appendix C.
\end{proof}

Notice that the proved range of $B^t$ means that $B^t > E_{max}$. So the constraint ($\ref{eq12}$) can be satisfied by $\eta^-((\sum_{i=1}^M N_i^t-U^t)^+-G^t) \le E_{max} \le B^t $.

$\mathbf{P5}(b)$ determines the CPU frequency allocation on server and the task-specific sale price for each vehicle. Due to the complexity of the original problem, we introduce a Lagrangian dual problem $\mathbf{\hat {P5}}(b)$. So the allocation and pricing strategy can be derived from a problem given by
\begin{align*}
    \mathbf{\hat {P5}}(b): \quad \max_{\mu_t,\nu_t} \quad \min_{\mathbf{f},\mathbf{g}} \quad \phi_t(\mathbf{f}_t^{server},\mathbf{g}(f_{i,t}^{server}), \mathbf{\mu}_t,\mathbf{\nu}_t) \\
\end{align*}
where 
\begin{align*}
\phi_t(\mathbf{f}_t^{server},\mathbf{g}(f_{i,t}^{server}), \mathbf{\mu}_t,\mathbf{\nu}_t) = \\ -\sum_{i=1}^M\beta_i^tR_i^t(\kappa\widetilde{B}^{t}\eta^-f_{i,t}^{server}+\frac{H}{f_{i,t}^{server}}g(f_{i,t}^{server}))\\
+\sum_{k=1}^K\nu_{k,t}(\sum_{i=1}^M \mathbbm{1}_{\{\beta_i^t>0\}}r_{i,k}-\Omega_k)
\end{align*}
and  $\nu_t = \{\nu_{1,t},\nu_{2,t},...\nu_{K,t}\}$, is the vector of Lagrangian multiplies and must satisfy
\begin{equation}
     \nu_{k,t} \ge 0 \quad  \forall k \in \mathcal{K}
\end{equation}

The dual problem is solved by using gradient projection method, and the Lagrangian multipliers are updated as following:
\begin{equation}
    \nu_{k,t}^{(n)} = [\nu_{k,t}^{(n-1)}-\ell(\sum_{i=1}^M \mathbbm{1}_{\{\beta_i^t>0\}}r_{i,k}-\Omega_k)]^+
\label{eq25}
\end{equation}

We display the whole decision process in $\mathbf{Algorithm} $ $\mathbf{\ref{algorithm1}}$, where the core idea is optimization of $\nu$ , $f$ and $g$, iteratively and alternatively.

\begin{algorithm}[htb] 
\caption{Framework of Solving $\mathbf{\hat {P5}}(b)$} 
\label{alg:Framwork} 
\begin{algorithmic}[1] 
\Require 
$\widetilde{B}^t$; 
\Ensure 
$f_{i,t}^{server}$, $g(f_{i,t}^{server})$; 
\label{code:fram:extract} 
\State Let all $\nu_{k,t}^{(0)}$ = $\nu_{max}$.
\label{code:fram:trainbase} 
\label{code:fram:add} 
\While {$\big|\phi_t|_{\nu_{t}^{(n)}}$ - $\phi_t|_{\nu_{t}^{(n-1)}}\big| \ge \varepsilon $ and $n < N$}
\State Compute new values of $\nu_t$ through (\ref{eq25}).
\State Let $f^* := \{\}$ and $g^* := \{\}$.
\For{$i=1$ to $M$}
\State Consider the three cases separately:
\If {$\Psi_{i,t}^{off} < \max\{\Psi_{i,t}^{loc},\Psi_{i,t}^{cld}\}$}
{
}
\State Select ${f}_{1}^*$ and $g_{1}^*$ that minimize $\Lambda_{i,1}$.
\ElsIf{$\Psi_{i,t}^{loc} < \max\{\Psi_{i,t}^{off},\Psi_{i,t}^{cld}\}$}{}
\State Select ${f}_{2}^*$ and $g_{2}^*$ that minimize $\Lambda_{i,2}$.
\Else
\State Select ${f}_{3}^*$ and $g_{3}^*$ that minimize $\Lambda_{i,3}$.
\EndIf
\State Let $\Lambda_{i,h} := \min \{\Lambda_{i,1},\Lambda_{i,2},\Lambda_{i,3}\}$.
\State $f^*.push\_back({f}_{i,h}^*)$ and $g^*.push\_back(g_{i,h}^*)$
\EndFor
\State $n = n + 1.$
\EndWhile \\
\Return $ f^*$, $ g^*$; 
\end{algorithmic} 
\label{algorithm1}
\end{algorithm}

We have derived three different optimization objectives under different cases, i.e. $\Lambda_1$, $\Lambda_2$, $\Lambda_3$. They are defined as below:
\begin{align*}
    \Lambda_{i,1} = & - R_i^t(\kappa\widetilde{B}^{t}\eta^-f_{i,t}^{server}+\frac{H}{f_{i,t}^{server}}g(f_{i,t}^{server}))\\
    &+\sum_{k=1}^K\nu_{k,t}(\sum_{i=1}^M  r_{i,k}-\Omega_k)\\
    \Lambda_{i,2} = & H(f_s^{local})^2(\frac{g(f_{i,t}^{server})}{f_{i,t}^{server}})^2 + \sum_{k=1}^K\nu_{k,t}( r_{i,k}-\Omega_k) \\
    &+(\kappa\widetilde{B}^{t}\eta^-f_{i,t}^{server}(f_s^{local})^2-R_i^tTH)\frac{g(f_{i,t}^{server})}{f_{i,t}^{server}}\\
    &-R_i^t T \kappa\widetilde{B}^{t}\eta^-f_{i,t}^{server}\\
    \Lambda_{i,3} = &-\sum_{k=1}^K\nu_{k,t}\Omega_k
\end{align*}

Interestingly, we can efficiently derive the solutions to all the three optimization problems due to the linearity or quadraticity of the objective functions. The solving procedure can be seen in Appendix D.

\begin{figure}[tb!]
  \centering
  \includegraphics[width=0.5\textwidth]{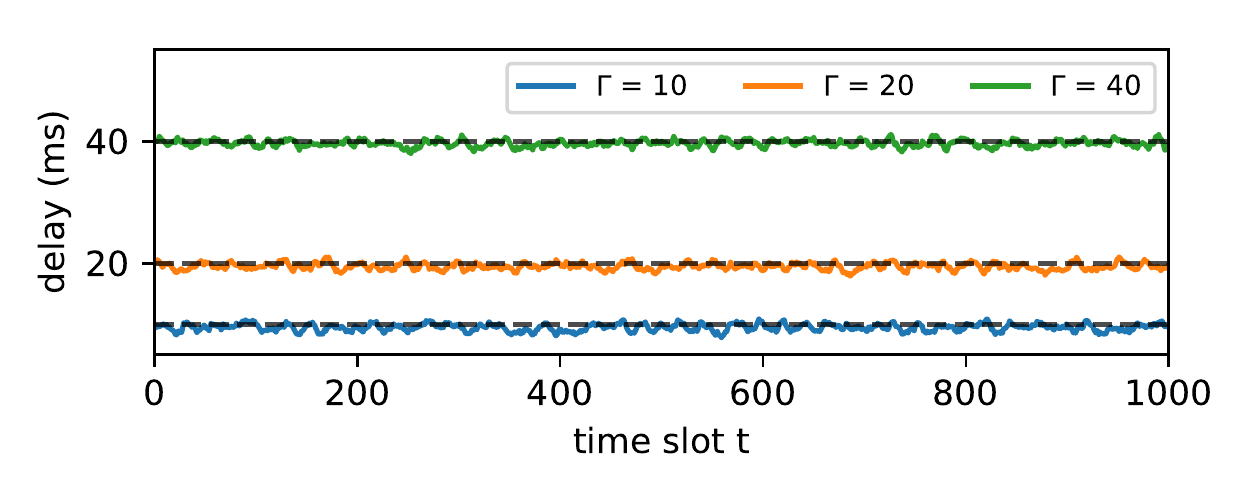}
   \vspace{-15pt}
  \caption{}
  \label{fig:delay_t}
\end{figure}

\begin{figure}[tb!]
  \centering
  \includegraphics[width=0.5\textwidth]{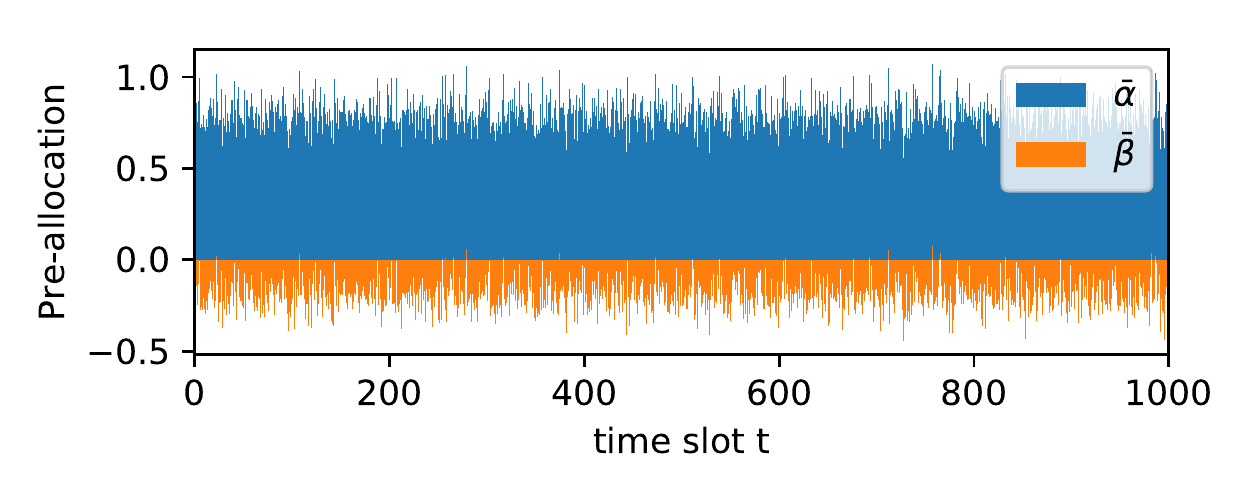}
   \vspace{-15pt}
  \caption{}
  \label{fig:allocation_t}
\end{figure}

\section{Numerical Results}
\begin{figure*}[tb!]
  \centering
  \includegraphics[width=1.0\textwidth]{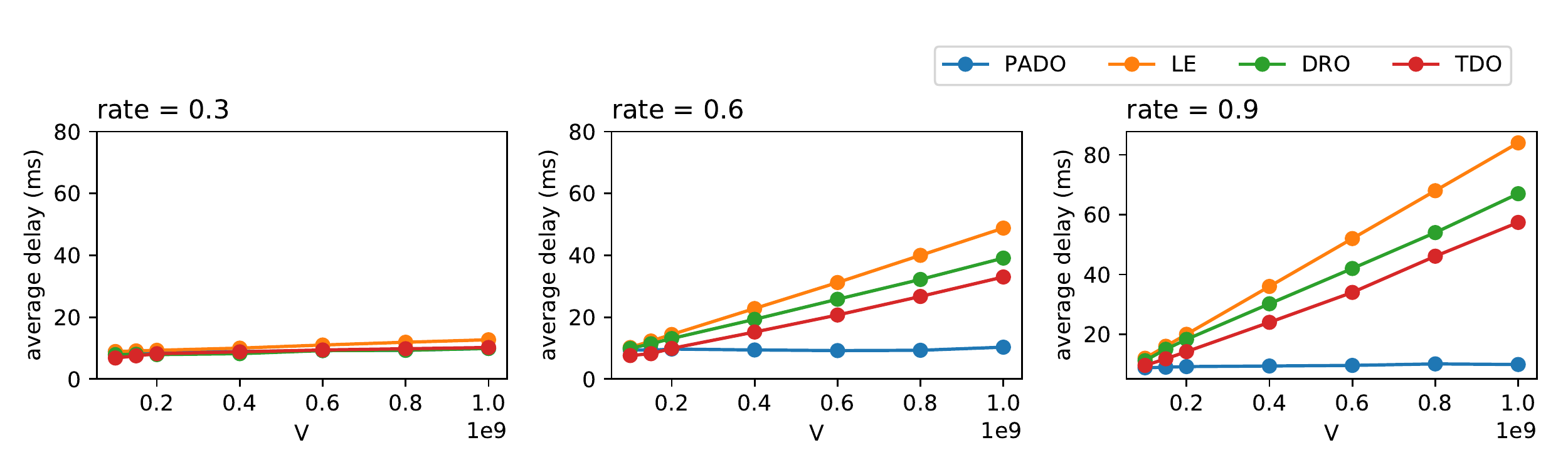}
   \vspace{-15pt}
  \caption{}
  \label{fig:delay_V}
\end{figure*}

In this section, we will verify the theoretical results derived in Section \uppercase\expandafter{\romannumeral4} and evaluate the performance of the proposed algorithm through simulations. We consider an MEC system with $|\mathcal{M}|$ = 50 randomly deployed mobile vehicles. We also set $\kappa = 10^{-28}$, $L = 1000$ cycles/bit and $f_i^{max} = 2$ GHz for all mobile vehicles. The locally generated computation task $R_{task}^t$ is assumed to be uniformly distributed within [10, 20] units, and each unit represents 1000 bit computation amount. The simulation results are conducted over 1000 consecutive time slots with slot length $\tau = 1$ ms. The control parameter $V$ is chosen extensively from $V = 10^{8}$ to $V = 10^{11}$.

We compare our proposed method with the existing paradigms listed below:

i) $\mathbf{Local}$ $\mathbf{Execution(LE)}$: No offloading happens in this scenario. All the tasks are executed locally, with local CPU frequency that maximizes $\mathbf{P2}$.

ii) $\mathbf{Dynamic}$ $\mathbf{Random}$ $\mathbf{Offloading(DRO)}$: The vehicles will stochastically offload part of its tasks to the MEC server. And the MEC Server accept these offloaded tasks as long as it has surplus computational resources. Otherwise, the tasks will be offloaded to the cloud.

iii) $\mathbf{Task}$ $\mathbf{Backlog}$ $\mathbf{Based}$ $\mathbf{Dynamic}$ $\mathbf{Offloading (TDO)}$: The vehicles can make decisions to locally execute the tasks or offload them. As a mainstream method in the preceding literature~\cite{du2018computation}, this approach focuses on maintaining the stability of a task backlog queue. For fair comparison, other settings of the framework are shared in our paper.

\subsection{Service Delay Performance}
Fig.~\ref{fig:delay_t} shows the delay evolution with different task deadlines throughout the testing time. The black dashed lines demonstrate the delay requirement of each kind of tasks. From the figure, we can observe that the average of delay is guaranteed in our method by stabilizing the delay queue. Fig.~\ref{fig:allocation_t} shows the corresponding pre-allocation strategy for the vehicle to stabilize its delay queue. It shows that our framework enables the vehicles to fully utilize the pre-set deadline requirement so that the delay can jitter close to the dashed line.

In Fig.~\ref{fig:delay_V}, we demonstrate the service delay performance for PADO algorithm and other three benchmark methods. We test the methods under different task arrival rates, which represent the severity of workload. 

It can be observed that when the task arrival rate is low, all the approaches can easily handle the situation by vehicles' local computation resource. In this situation, PADO will allocate its resources and fulfill the delay requirement more wisely. When the task arrival rate becomes high,  for all the benchmarks, the average delay increases linearly with V and becomes unbounded when V goes to infinity. However, with our pre-allocation framework PADO, the service delay can be controlled beforehand via a pre-defined parameter. This property is well preserved especially when the control parameter $V$ is relatively small, which means the vehicle attaches more importance to the service delay.

\begin{figure}[tb!]
  \centering
  \includegraphics[width=0.45\textwidth]{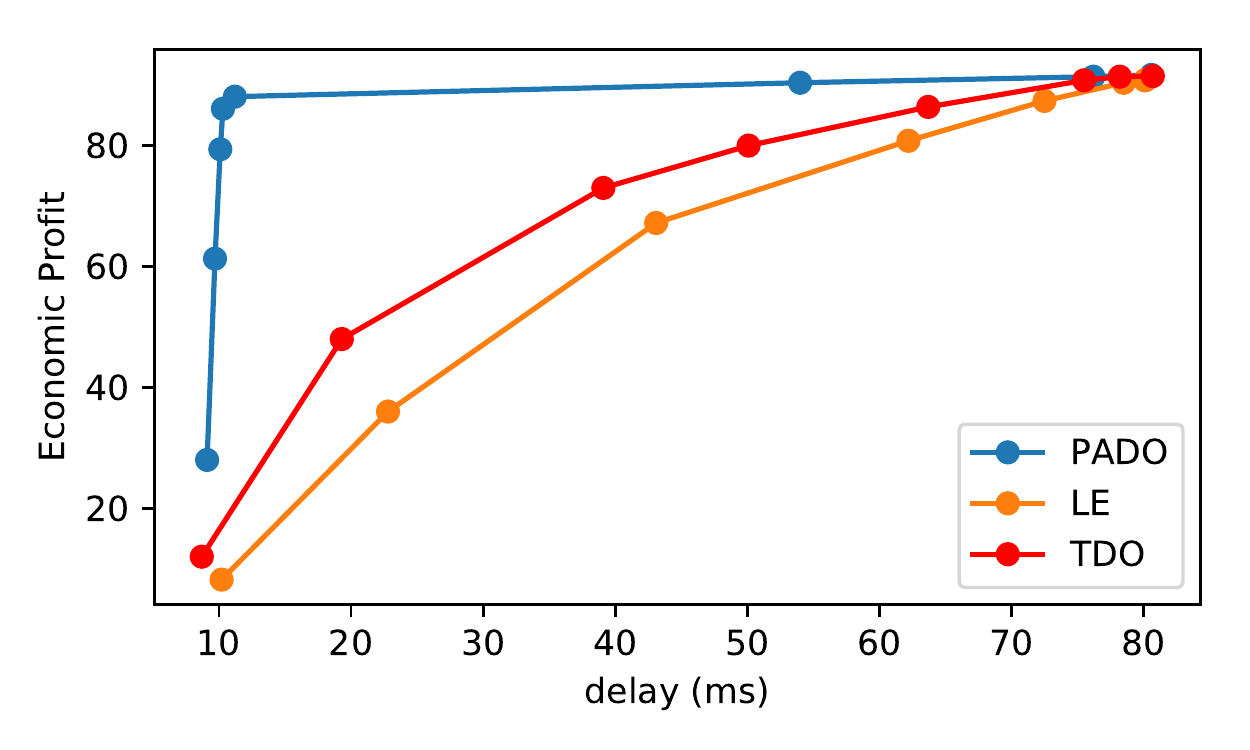}
   \vspace{-15pt}
  \caption{}
  \label{fig:tradeoff}
\end{figure}

\subsection{Economic Profit of Vehicles}
Other than the delay requirement, our framework also focuses on economic profit for each vehicle. The economic profit includes minimizing local energy consumption, offloading payment and dropping loss. So the PADO framework can be summarized as searching for the highest economic profit that satisfies task's delay constraint. From Fig.~\ref{fig:tradeoff}, we observe that methods under Lyapunov optimization framework have tradeoff between task delay and execution expenditure. For Local Execution Strategy (LE), when the control parameter $V$ is low, it will allocate more local CPU resources to lessen its execution delay, so the economic profit is low. On the other hand, when $V$ goes larger,  the economic profit rises for less local energy consumption, but with the delay rises correspondingly. For TDO and PADO, this process is more complex as the vehicles can decide their offloading strategy by tuning $\alpha$ and $\beta$ and they both show a similar tradeoff with LE. The difference between TDO and PADO is that our proposed PADO has a phase that stabilizes the delay but largely varies in economic profit. As stated above, vehicles in this phase can fully utilize the deadline and optimizes its economic profit.  We observe that the three methods converge to the same point in Fig.~\ref{fig:tradeoff}. It means when V is ultra-large, more emphasis is put on the economic profit. The most efficient way to lower expenditure is to locally execute the tasks regardless of delay. However, such $V$ should be avoided in real world application as the offloading strategy does not take advantage of the service provided by MEC Server.

\begin{figure}[tb!]
  \centering
  \includegraphics[width=0.45\textwidth]{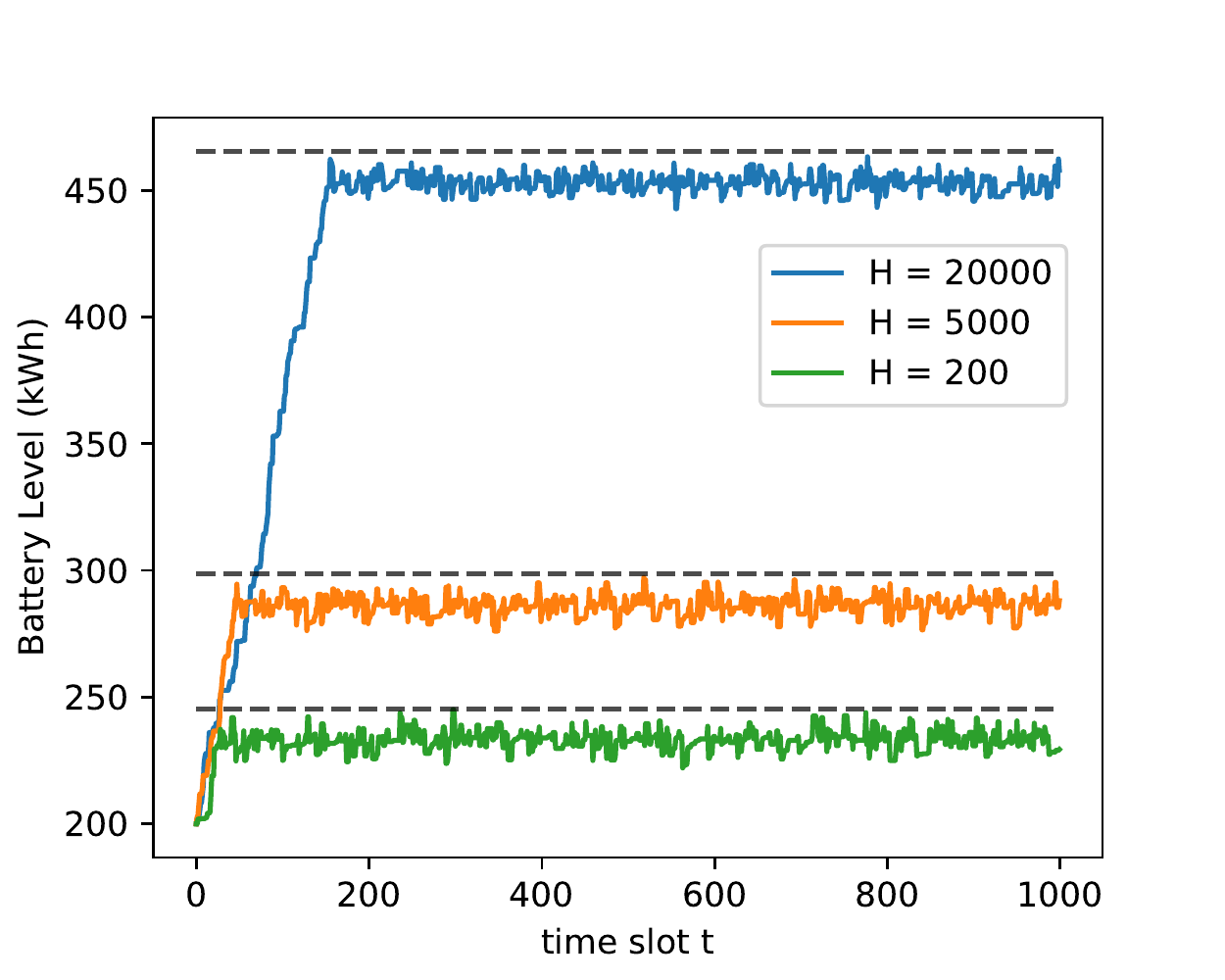}
   \vspace{-15pt}
  \caption{}
  \label{fig:battery}
\end{figure}

\subsection{Performance on MEC Server}
Our PADO framework also pays attention to the performance on the service providers, a.k.a the MEC Server. One of the main concerns of MEC Server is the battery charge/discharge management. First, 
Fig.~\ref{fig:battery} illustrates the battery level evolution process under control parameter $H$ ranging from 200 to 20000. It directly shows the stability of battery level, which has been theoretically proven in Appendix~\ref{AppC}. The battery is initially charged with 200 kWh electricity power. We can see that for different $H$, the battery levels first increase linearly and then jitter within a fixed range. Throughout the time slots, the battery level never crosses the upper bound $\theta - H \chi^t / \eta^- +\eta^+C^{max}$, which is represented by the balck dashed lines in Fig~\ref{fig:battery} for different $H$ respectively. When $H$ becomes larger, the stable level of battery also increases. It means the battery size should be expanded for increased $H$.

\begin{figure}[tb!]
  \centering
  \includegraphics[width=0.45\textwidth]{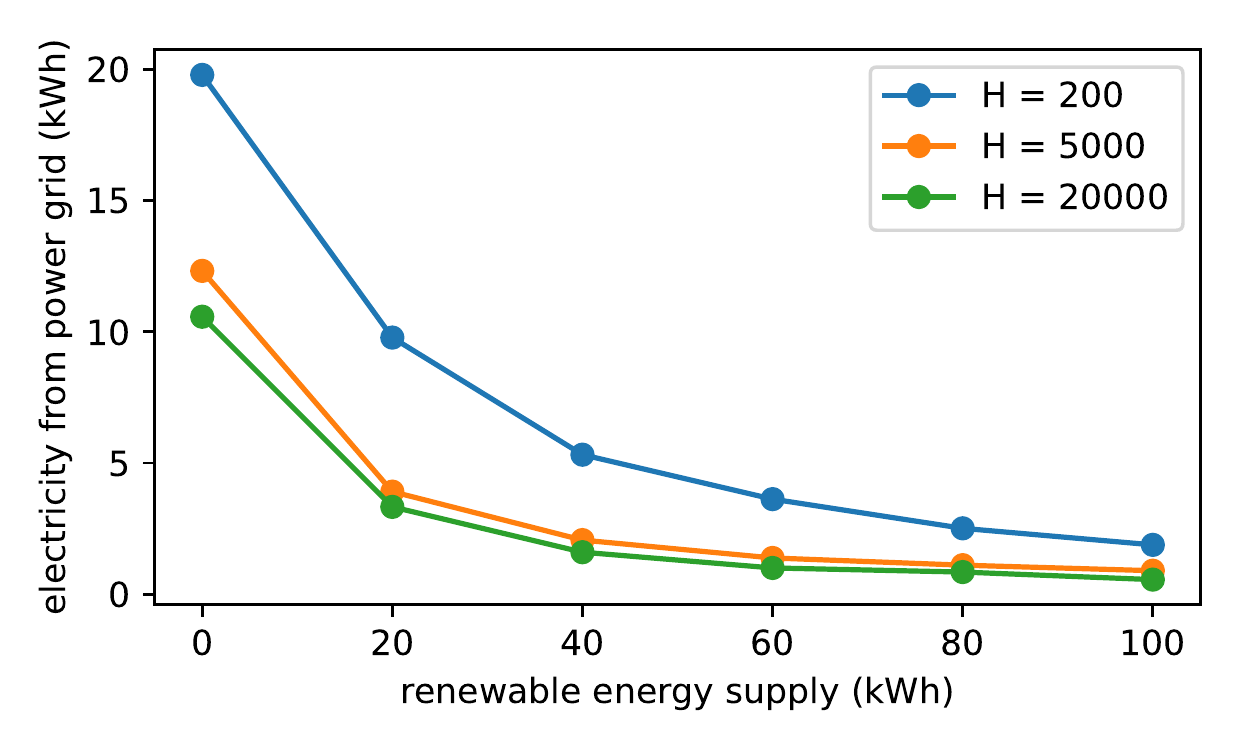}
   \vspace{-15pt}
  \caption{}
  \label{fig:renew}
\end{figure}

Fig~\ref{fig:renew} shows the energy management strategy of the MEC Server. One of the most frustrating problems in using renewable energy as power supply is its unstable generation. We can see that our PADO framework enables the MEC Server to automatically change its energy source. When the renewable energy supply is low, the energy purchased from power grid would be high to maintain its battery level and provide service to vehicles. Another property is that when the renewable energy supply reduces, the overall energy supplement (from power grid and renewable source) will also decline. We also compare the energy management strategy under different $H$. When $H$ is large, the MEC Server lays more weight to the expenditure, so the purchase amount drops correspondingly.

In our MEC settings, the MEC Server has limited computation resources, which means it cannot provide service for all the offloading requests. The server may accept or reject some of them by tuning its sale price to maximize its own profit. Fig.~\ref{fig:resource} shows the relation between the quantity of computation resources that an MEC Server owns and its transaction price with vehicles. As the MEC Server has more computation resources, the average price\footnote{The average price means the average transaction price for one unit of computation task.} decreases to attract more offloading requests from vehicles. Indeed the loss of income in unit price is worthwhile for the MEC Server, since the server can earn it back by expanding the customer group. So we also observe an increase in the total revenue. The unit price will no longer decrease When the MEC Server has abundant computation resources. It can be explained by the reason that further decrease in price will not generate higher overall profit for the server.

\begin{figure}[tb!]
  \centering
  \includegraphics[width=0.45\textwidth]{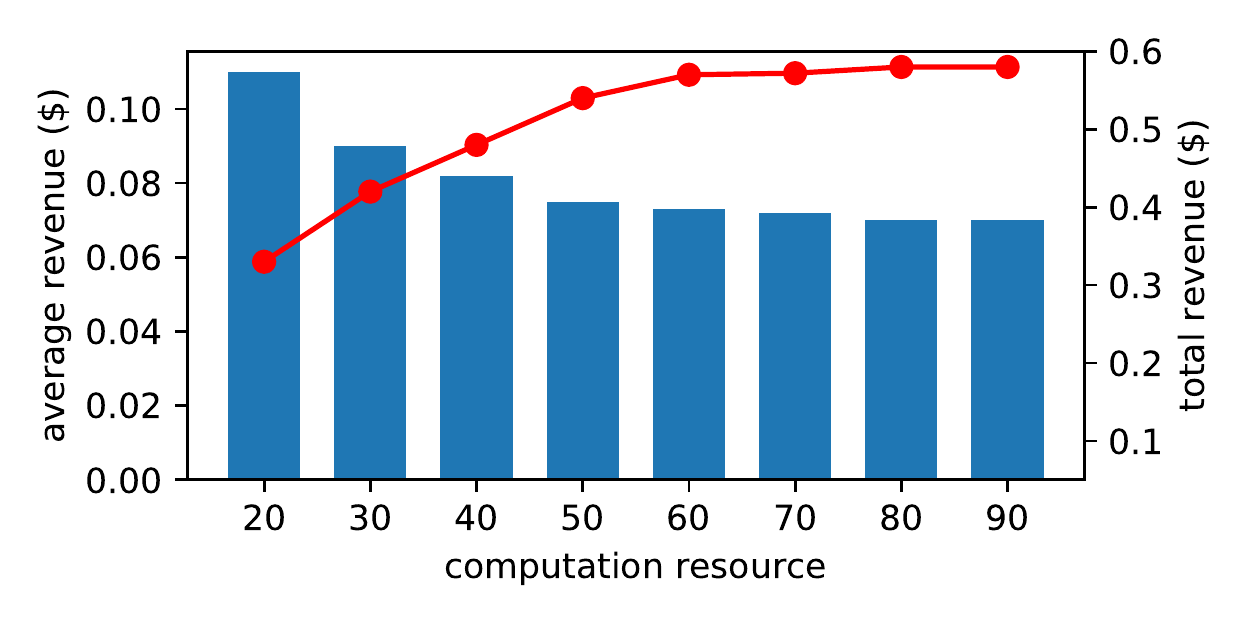}
   \vspace{-15pt}
  \caption{}
  \label{fig:resource}
\end{figure}

\begin{appendices}

\section{}
We prove the constraints for $B^t$ by induction. First, if $\eta^-\theta-H\chi^t < B^t < \theta\eta^--H\chi^t+\eta^+C^{max}$, it denotes that battery has abundant energy, so the server choose to discharge its energy from battery. So $B_i^{t+1} < B_i^{t}$. Moreover, $G^t > (\sum_{i=1}^M N_i^t-U^t)^+ - E_{max}$. Hence $B^{t+1} > B^t -E_{max} + \eta^+C^t > \eta^-\theta-H\chi^t - E^{max} > E^{max}$ according to our definition of $\theta$ in (\ref{eq15}). When $E_{max} < B^t < \eta^-\theta-H\chi^t$, $B_i^{t+1}=B_i^t+\eta^+C^t < \eta^-\theta-H\chi^t + C^{max}$. Also, $B_i^{t+1} \ge B_i^t > E_{max}$. Therefore, we have $B^t \in  [E_{max},\theta\eta^--H\chi^t+\eta^+C^{max}]$.
\label{AppC}
\end{appendices}

\bibliographystyle{abbrv}
\bibliography{survey}

\begin{thebibliography}{10}

\bibitem{burd1996processor}
T.~D. Burd and R.~W. Brodersen.
\newblock Processor design for portable systems.
\newblock {\em Journal of VLSI signal processing systems for signal, image and
  video technology}, 13(2-3):203--221, 1996.

\bibitem{cordeschi2014reliable}
N.~Cordeschi, D.~Amendola, and E.~Baccarelli.
\newblock Reliable adaptive resource management for cognitive cloud vehicular
  networks.
\newblock {\em IEEE Transactions on Vehicular Technology}, 64(6):2528--2537,
  2014.

\bibitem{du2018computation}
J.~Du, F.~R. Yu, X.~Chu, J.~Feng, and G.~Lu.
\newblock Computation offloading and resource allocation in vehicular networks
  based on dual-side cost minimization.
\newblock {\em IEEE Transactions on Vehicular Technology}, 68(2):1079--1092,
  2018.

\bibitem{liu2014effective}
Y.~Liu and M.~J. Lee.
\newblock An effective dynamic programming offloading algorithm in mobile cloud
  computing system.
\newblock In {\em 2014 IEEE Wireless Communications and Networking Conference
  (WCNC)}, pages 1868--1873. IEEE, 2014.

\bibitem{mao2016dynamic}
Y.~Mao, J.~Zhang, and K.~B. Letaief.
\newblock Dynamic computation offloading for mobile-edge computing with energy
  harvesting devices.
\newblock {\em IEEE Journal on Selected Areas in Communications},
  34(12):3590--3605, 2016.

\bibitem{qiu2017augmented}
H.~Qiu, F.~Ahmad, R.~Govindan, M.~Gruteser, F.~Bai, and G.~Kar.
\newblock Augmented vehicular reality: Enabling extended vision for future
  vehicles.
\newblock In {\em Proceedings of the 18th International Workshop on Mobile
  Computing Systems and Applications}, pages 67--72. ACM, 2017.

\bibitem{rabaey2002digital}
J.~M. Rabaey, A.~P. Chandrakasan, and B.~Nikolic.
\newblock {\em Digital integrated circuits}, volume~2.
\newblock Prentice hall Englewood Cliffs, 2002.

\bibitem{sasaki2016vehicle}
K.~Sasaki, N.~Suzuki, S.~Makido, and A.~Nakao.
\newblock Vehicle control system coordinated between cloud and mobile edge
  computing.
\newblock In {\em 2016 55th Annual Conference of the Society of Instrument and
  Control Engineers of Japan (SICE)}, pages 1122--1127. IEEE, 2016.

\bibitem{vu2012real}
A.~Vu, A.~Ramanandan, A.~Chen, J.~A. Farrell, and M.~Barth.
\newblock Real-time computer vision/dgps-aided inertial navigation system for
  lane-level vehicle navigation.
\newblock {\em IEEE Transactions on Intelligent Transportation Systems},
  13(2):899--913, 2012.

\bibitem{wang2017computational}
W.~Wang and W.~Zhou.
\newblock Computational offloading with delay and capacity constraints in
  mobile edge.
\newblock In {\em 2017 IEEE International Conference on Communications (ICC)},
  pages 1--6. IEEE, 2017.

\bibitem{xu2013survey}
Y.~Xu and S.~Mao.
\newblock A survey of mobile cloud computing for rich media applications.
\newblock {\em IEEE Wireless Communications}, 20(3):46--53, 2013.

\bibitem{yang2016distributed}
B.~Yang, J.~Li, Q.~Han, T.~He, C.~Chen, and X.~Guan.
\newblock Distributed control for charging multiple electric vehicles with
  overload limitation.
\newblock {\em IEEE Transactions on Parallel and Distributed Systems},
  27(12):3441--3454, 2016.

\bibitem{zaharia2010spark}
M.~Zaharia, M.~Chowdhury, M.~J. Franklin, S.~Shenker, and I.~Stoica.
\newblock Spark: Cluster computing with working sets.
\newblock {\em HotCloud}, 10(10-10):95, 2010.

\bibitem{zhang2018energy}
G.~Zhang, W.~Zhang, Y.~Cao, D.~Li, and L.~Wang.
\newblock Energy-delay tradeoff for dynamic offloading in mobile-edge computing
  system with energy harvesting devices.
\newblock {\em IEEE Transactions on Industrial Informatics}, 14(10):4642--4655,
  2018.

\end{thebibliography}

\end{document}